\documentclass{article}
\usepackage{spconf,amsmath,amsthm,graphicx,cite}

\newtheorem{claim}{Claim}

\title{Social Bubbles and Superspreaders: Source Identification for Contagion Processes on Hypertrees}

\name{Sam Spencer and Lav R.\ Varshney\thanks{This work was supported in part by NSF grants OAC-1639529 and ECCS-2033900, and the Center for Pathogen Diagnostics through the ZJU-UIUC Dynamic Engineering Science Interdisciplinary Research Enterprise (DESIRE).}}
\address{Coordinated Science Laboratory and	Department of Electrical and Computer Engineering\\
	University of Illinois at Urbana-Champaign, Urbana, IL 61801, USA}

\begin{document}

\maketitle

\begin{abstract}
Previous work has shown that for contagion processes on extended star networks (trees with exactly one node of degree $>$ 2), there is a simple, closed-form expression for a highly accurate approximation to the maximum likelihood infection source. Here, we generalize that result to a class of hypertrees which, although somewhat structurally analogous, provides a much richer representation space. In particular, this approach can be used to estimate \emph{patient zero} sources, even when the infection has been propagated via large group gatherings rather than person-to-person spread, and when it is spreading through interrelated \emph{social bubbles} with varying degrees of overlap. In contact tracing contexts, this estimator may be used to identify the source of a local outbreak, which can then be used for forward tracing or for further backward tracing (by similar or other means) to an upstream source.
\end{abstract}


\section{Introduction}
\label{sec:introduction}

Localizing sources of spreading processes in networks has become a well-studied inference problem in information theory and signal processing \cite{ShahZamanIT11,LuoTL2013,ZhouJV2019}.
As part of COVID-19 pandemic response, there has been a strong focus on (forward) contact tracing to determine who might have been exposed to the pathogen.  Due to the clustered nature of this disease spread, however, there has been recent interest in backward contact tracing to determine the source of a pathogen to then facilitate forward tracing \cite{Tufekci2020,EndoKMAFK2020,BradshawAHLE2020}.  This is especially effective since it is estimated that 80 percent of infections arise from only about 10 percent of cases \cite{Bi2020,Adam2020}---so-called \emph{superspreading events}.  

Distinct from the clustered nature of disease spread, there is greater clustering in social networks themselves as people form \emph{social bubbles} due to contact adaptation via social distancing \cite{SahnehVMS2019}.  Such bubbles may include families living together, co-workers in close quarters, friends in social gatherings, or children in daycare. Transmission within such settings may well occur far more easily and quickly than through less direct, sustained interactions.  In practice, people are interacting in partially overlapping bubbles, e.g.\ daycare students are also in family groups. As such, spreading models that capture group-based transmission may be more insightful than ones based on person-to-person transmission, so as to localize COVID-19 superspreading events.   

To model pathogen spread via groups, one may use bipartite graphs \cite{Frost2007} where people are one kind of node and associations another.  As an alternative, we consider spread over \emph{hypergraphs} that correspond to such bipartite graphs, where nodes are people and hyperedges are associations.  As summarized in a very recent review paper \cite[Sec.~7.1.2]{Battiston2020}, there is some study of spreading process dynamics in hypergraphs \cite{BodoKS2016,SuoGS2018,ArrudaPM2020,AntelmiCSS2020}, but no study on source localization. 

Since infection source localization in general networks is mathematically difficult, one must make some simplifying assumptions to define a more tractable problem.  In our case, we  extend the model of susceptible-infected (SI) spreading over extended star networks  \cite{SpeSri16} to an analogous class of extended star \emph{hypertrees}. For the situations described by our model, we define a computationally-efficient  approximation to the maximum likelihood estimate (MLE) of the hyperedge containing the source of the infection relative to the particular infected set in an observed snapshot of data.  We then prove the approximation is close to MLE (in the sense of \cite{SpeSri16}), and also give simulations demonstrating efficacy of the estimator.  

Beyond COVID-19 and similar pandemics, our approach also captures broadcasting information in a one-to-many manner, rather than person-to-person communication as modeled in traditional rumor spreading \cite{ShahZamanIT11}.

\section{Source Localization for Graphs}
\label{sec:background}
Let us review the standard SI compartment model of infection spreading on graphs, with edge-based propagation in continuous time. That is, at any particular point in time, nodes are either \emph{susceptible} or they are \emph{infected} and capable of infecting susceptible individuals. All nodes begin as susceptible, except for one (the \emph{source}), which becomes infected through some outside action. Subsequently, if a susceptible node shares an edge with an infected neighbor, the infection will traverse that edge and infect the susceptible node with a waiting time that is exponentially distributed with mean $T$. Once infected, a node remains that way indefinitely. 
Due to the  memorylessness property of the system, at any given time the next infection is equally likely to occur along any outgoing edge from the current infected set. 

Assume there is a single source. Source localization in a graph is then, given an observed infection pattern (the subgraph of nodes infected at some point in time), finding the MLE of the source node giving rise to that infection pattern.  

Let an \emph{extended star network} be a tree with exactly one node of degree $ >  2$: this center node is denoted $O$.  The $m$ neighbors of this node give rise to $m$ \emph{arms}---chains of successive nodes of degree $2$. The nodes of each arm are numbered starting with $1$ (for the node adjacent to $O$) and increasing from there.

If the contagion were confined to a single arm, source localization reduces to a problem on a line graph, and the MLE is well-known to be the midpoint of the infection (in fact, for a uniform prior, the likelihood function follows a binomial distribution  
\cite{ShahZamanIT11,SpeSri15}).  Therefore, we assume that any nontrivial infection pattern will comprise $O$, along with the closest $k_i$ nodes along each arm $i$.
For an extended star network, there is a simple, closed-form expression for a highly accurate approximation to the MLE  \cite{SpeSri16}. 
It is, along the longest arm:
\begin{align}
\label{eq:old_est}
\ell= \frac{k_1-\frac{\sum_{i=2}^m k_i}{m-1}}{2}\mbox{.}
\end{align}
Further analysis in \cite{KesSpe19} tightened the approximation bound for this estimator. Note that estimator \eqref{eq:old_est} is scale-invariant---all arm lengths $k_i$  can be scaled by a constant, and the resulting value of $\ell$ will be scaled by the same constant. 

\section{Spreading Process on Hypergraphs}
\label{sec:model}

Having an understanding of spreading and source localization on graphs, now we turn to hypergraphs.  

A \emph{hypergraph} is a generalization of a graph, consisting of a set of nodes and hyperedges. Unlike in a graph, where an edge can only connect two nodes, a hyperedge can connect two or more nodes. For present purposes, we further define a \emph{hypertree} (cf.~junction trees in statistical inference \cite{QiM2004}) as a hypergraph that inherits the conventional properties of trees---that is, it is connected, undirected, and acyclic.

For source localization, our hypertree consists of a set of nodes representing all individuals under study, and a set of hyperedges representing the various non-distanced groupings in which these nodes exist. We  denote these hyperedges as \[
\{E_0, E_{1,1}, 
\dots,E_{1,k_1}, E_{2,1},\dots, E_{2,k_2},\dots, E_{m,1}, \dots, E_{m,k_m}\}
\]
such that for all $i$, $E_0\cap E_{i,1} \neq \emptyset$; for all $i,j$, $E_{i,j}\cap E_{i,j+1} \neq \emptyset$; and 
all other intersections are empty.

Intuitively, $ E_0 $ represents a superspreading event, with its members being individuals in attendance. 
The remaining hyperedges are arranged in $ m $ \emph{arms} of lengths $ \{k_1,\dots, k_m\}$. The first hyperedge of each arm has an 
overlap with $ E_0 $, as do successive neighboring hyperedges within each arm.

To develop a hypergraph SI model, we consider each hyperedge to be in either a  susceptible or infected state, since we consider close and continuing contact among members of a hyperedge.  Once infected, a hyperedge remains that way indefinitely. We recognize this means a node can be a member of both a susceptible hyperedge and an infected hyperedge at the same time, but 
such an ambiguity is at worst only temporary.  If necessary, ambiguities may be resolved as desired (presumably in accordance with an appropriate physical model), as long as it is done in a consistent manner. 

Spreading is still modeled in continuous time.  When $ v $ members of a susceptible hyperedge are also members of a different, infected hyperedge, the entire set of nodes in the susceptible hyperedge becomes infected with a waiting time that is exponentially distributed with mean $ T/v $. To see why this makes sense, suppose we were to consider infections as taking place over simple edges (hyperedges of size 2), then we would model waiting time as exponentially distributed with mean $ T $.  Since there are $ v $ nodes in the overlap, we stipulate that transmission over any one of them is sufficient to convey infected status from one hyperedge to the other (since nodes within a single hyperedge are assumed to have close, continuing contact).  Accordingly, the appropriate transmission time is distributed as the lesser of the two node-based transmission times, which provides the desired result (exponentially distributed with mean $ T/v $).

\section{Source Localization for Hypergraphs}
\label{sec:generalization}
How does the source localization result for extended star graphs from Sec.~\ref{sec:background} extend to our hypergraph setting in Sec.~\ref{sec:model}? If we consider our hyperedges to be nodes in a graph, then we also have an infection spreading over an extended star network.  However, the difference is that the propagation times are now no longer uniformly exponentially distributed with mean $ T $, but instead vary with each hop.

We must determine the effect of shorter, variable transmission times on the growth of the infected set. Consider a case where two arms are growing simultaneously.  The first arm consists of edges with waiting time $ T $.  The second arm consists of edges with waiting time $ T/v $. Intuitively, it is clear that, over time, the second arm will grow at $ v $ times the rate of the first arm.  Alternatively, if we were to look at a snapshot of the second arm and find it to have length $ k $, the maximum likelihood estimate for the length of the first arm would be $ k/v $. Accordingly, we conclude that we should apply a weight corresponding to this rate to each hyperedge transmission in our model, so that the weighted length of an arm of hyperedges represents, in a typical sense, the equivalent unweighted length of an equivalent set of simple edges. Thus, since a simple edge hop is considered to have length 1, we will denote a hyperedge hop with $ v $ overlapping nodes to have equivalent length $ 1/v $.  We now show this more formally.

\begin{claim}
Under the spreading dynamics described above, over an equivalent period of time, the probability of observing $ vx $ hops on an arm of hyperedges with successive overlap $ v $ is equal to the probability of observing $ x $ hops on a singly-connected arm.
\end{claim}
\begin{proof}
Consider two related structures. Let the first be an extended star network with $ v+1 $ arms of singly-connected simple edges. Let the second be a hypertree with an origin node, one arm of singly-connected simple edges, and one arm of hyperedges with successive overlap $ v $. In each case, there are always $ v+1 $ possibilities for the next step in the propagation of the infection. Since the waiting times for each step are exponentially distributed, each of those distributions are memoryless, and thus each of the choices for the potential next steps are equiprobable. Let us label a step along the simple edge arm as choice 1, and the remaining possible steps, all along the hyperedge arm, as choices $2, \dots, v+1$. Now, suppose we have $ K $ total steps taken.  Then the probability for a given breakdown $ \{k_1, k_2, \dots, k_{v+1}\} $ of each of the choices can be written as
\begin{align}
P(k_1, k_2,\dots, k_{v+1}) = \frac{K!}{k_1 ! k_2 !\dots k_{v+1} !(v+1)^K }
\end{align}        
From this expression, it is clear that the maximum likelihood is attained when all of the $ k_i $ values are equal (let us call this value an unsubscripted $k$).  To see that this is the case, consider what happens if we replace one of the $ k $ values with $  k+1 $, and another with $ k-1 $. Then the numerator of the probability remains unchanged, but the denominator loses a factor of $ k $ while gaining a factor of $ k+1 $. This net increase in the denominator implies a decrease in the overall probability.  Since any perturbation from equal $k$ values can be decomposed as a sequence of similar steps (removing smaller factors from the denominator and replacing them with larger factors), the resulting probability would be strictly less.
Now, since each of the choices $\{2,\dots,v+1\}$ imply a step along the hyperedge arm, we conclude that in a maximum likelihood growth pattern, this arm would grow at $ v $ times the rate of a simply-connected arm.
\end{proof}
\begin{claim}
If a simple edge hop is considered to have length $1$, a hyperedge hop with $ v $ overlapping nodes can be considered to have equivalent length $  1/v $.
\end{claim}
\begin{proof}
Denote the correct equivalent length of the hyperedge hop by $ z $. Now, if we scale up the problem by a factor of $ v $, then we have $ v $ single hops and $ v^2 $ hops along the hyperedge.  By the assumption, the total length of the $ v $ single hops is $ v $, and since they are equivalent, the $ v^2 $ hyperedge hops must total to the same.  Therefore, each hyperedge hop can be considered to have equivalent length $ 1/v $.  
\end{proof}

Therefore, let $v_{i,1}=|E_{i,j} \cap E_{i,j-1}|$ and  $v_{i,j} =|E_{i,j} \cap E_{i,j-1} |$ for $ j>1 $. Then let $ w_i= \sum_{j=1}^{k_i} 1/v_{i,j} $ . Now, by weighting each transmission and arm length accordingly, and applying the formula from the edge-based model, we can state our resulting maximum likelihood hyperedges are given by the following.  That is to say, this is the proposed estimator.
\begin{claim}
Under the spreading dynamics given above, the maximum likelihood source for the contagion process is well-approximated (in the same sense as in \cite{SpeSri16}) by
\begin{align}
\label{eq:est}
E_{1,l} : \sum_{j=1}^l \frac{1}{v_{1,j}} &= \frac{w_1-\frac{\sum_{i=2}^m w_i}{m-1}}{2} \mbox{.}
\end{align}
\end{claim}
\begin{proof}
Substitute the weighted edge lengths into the maximum likelihood expression from \cite{SpeSri16}.
\end{proof}

\section{Sensitivity to Noisy Data}
\label{sec:sensitivity}
In most real-world settings, data is subject to errors.  We consider a few classes of errors, and evaluate their impact on the source estimate. Since results will depend on the weights of the links between successive nodes in the vicinity of the optimal ML estimate, they are given as order results. We  assume the overlaps $ w_i $ are distributed such that $ E[(1/w_i )] $ is finite.\vspace{1mm}

\noindent {\bf Missing Arm (Non-longest)}
If we miss an entire arm by not observing one of the outgrowths of the superspreader event (e.g.\ at a protest, without attendance information), the impact is likely to be small, as long as the equivalent length of that arm is typical of the other non-longest arms.  Specifically, in the summation in the numerator of \eqref{eq:est}, we will lose one of the $w_i $ values in the calculation of the average length of the non-longest arms. This error would be zero mean, standard deviation $ \Theta(\sigma/(m-1)) $, where $ \sigma $ represents the standard deviation of the weighted lengths of the non-longest arms.\vspace{1mm}

\noindent {\bf Missing Step (Non-longest Arm)}
If a single step is lost from a non-longest arm, the corresponding $ w_i $ will be reduced by the weighted length of that hop, which will then reduce the computed average weighted length of the non-longest arms. This will have an $ \Theta(1/m) $ impact in an outward direction along the longest arm.\vspace{1mm}  

\noindent {\bf Missing Arm (Longest)}
This error is the most serious.  Since we do no have the arm on which the MLE lies, the ML source node will definitely be lost entirely, and a source on the next longest arm (or at the center node $ O $) will be chosen instead.\vspace{1mm}

\noindent {\bf Missing Step (Longest Arm)}
This error is more significant than if it were on a different arm.  Since the reduction is now in $ w_1 $, it is no longer divided by $ m-1 $, and due to the sign change, it would have an impact of $ \Theta(1) $ in an inward direction along the longest arm, rather than $ \Theta(1/m) $  outward. 
This also assumes that this missing step does not change which arm appears to be the longest.  If it did, the source estimate would now be on a different arm entirely, a 
nonlinear estimation-theoretic 
threshold phenomenon \cite{VanTrees1968}.

\section{Numerical Experiments}
\label{sec:testing}
We examine increasingly more realistic but less general random infection patterns, and compare results from our estimator against those from a detailed time-domain simulation. All random choices are uniform. The random patterns chosen  include:\vspace{-2mm}
\begin{enumerate}
\item \emph{Unconstrained}:  There are a random number of 2--6 arms, each of random length between 1--50, and each hyperedge overlap size is random  between 1--6.\vspace{-2mm}

\item \emph{Constrained}: The number of arms is random (2--6) as before, but now all arms but one have random length (11--30), while a randomly chosen arm has random length (31--50). The overlap sizes are randomly selected (1--6). 
There is a distinct longest arm, but since we are interested in the weighted length rather than the number of overlaps per se, this may not carry over.\vspace{-2mm} 

\item \emph{Typical}: This 
setting reflects spread patterns that are \emph{typical} in the information-theoretic sense. The structure begins similarly to the constrained case: 2--6 arms, one of length 31--50, the rest of length 21--30, overlap size 1--6.  Now, we compute the weighted lengths of each arm. Arms with the smallest and largest weighted lengths remain unchanged.  For each remaining arm, we begin at the center and work our way outwards, computing a running total of the weighted length as we go. Once we surpass the smallest weighted length, the remaining part of the arm is truncated.  This process ensures the non-longest arms all have similar weighted lengths, as expected for a typical growth pattern.
\end{enumerate}
We also consider all three models, but with unit overlap (called \emph{single} rather than \emph{multiple}). This mimics simple graphs to an extent. The difference is that in simple graphs, the infection originates at a single node and spreads in all directions from there. 
In our hypertree model, the infection begins within a hyperedge, and the initial spread is through overlaps with adjacent hyperedges.  This means  the effective length of the longest arm is arguably one hop longer than in the simple graph case. 
We examine the effect of adding a small offset to the $w_1$ term in \eqref{eq:est} to counter this difference.

In each case, we run 1000 trials using the chosen parameters and model, and compute when our estimator selects a different arm than the time-domain simulation (including when one or the other choose the central hub, called ``Arm 0''), as well as when the two methods give different node results (given that arm selections match).  In addition to these two error rates, we measure the average hops of node errors, and the positivity rate for the errors (how often the index of the estimated node is larger than the index from the simulation) to identify systematic bias.  Results without offsets are in Table~\ref{tab:results}, whereas those with small offsets are in Tables~\ref{tab:results2}--\ref{tab:results5}.

\begin{table}
    \centering
    \begin{tabular}{|l|r|r|r|r|}
    \hline
         & \footnotesize{arm error} &   \footnotesize{node error} & \footnotesize{error size} & \footnotesize{positivity} \\ \hline
        \footnotesize{unconstr.~(sing.)} & $0.60\%$ & $9.40\%$ & $1.00$ & $88.20\%$ \\ \hline
        \footnotesize{unconstr.~(mult.)} & $2.10\%$ & $33.20\%$ & $1.22$ & $12.60\%$ \\ \hline
        \footnotesize{constr.~(sing.)} & $0.00\%$ & $4.80\%$ & $1.00$ & $97.90\%$ \\ \hline
        \footnotesize{constr.~(mult.)} & $0.80\%$ & $20.70\%$ & $1.01$ & $85.90\%$ \\ \hline
        \footnotesize{typical (sing.)} & $0.00\%$ & $0.40\%$ & $1.00$ & $0.00\%$ \\ \hline
        \footnotesize{typical (mult.)} & $0.00\%$ & $25.90\%$ & $1.00$ & $97.70\%$ \\ \hline
    \end{tabular}
    \caption{Estimator Performance (no offset)}
    \label{tab:results}
\end{table}

\begin{table}
    \centering
    \begin{tabular}{|l|r|r|r|r|}
    \hline
         & \footnotesize{arm error} &   \footnotesize{node error} & \footnotesize{error size} & \footnotesize{positivity} \\ \hline
        \footnotesize{unconstr.~(sing.)} & $0.10\%$ & $10.00\%$ & $1.00$ & $82.00\%$ \\ \hline
        \footnotesize{unconstr.~(mult.)} & $2.80\%$ & $37.60\%$ & $1.24$ & $4.70\%$ \\ \hline
        \footnotesize{constr.~(sing.)} & $0.30\%$ & $5.40\%$ & $1.00$ & $96.30\%$ \\ \hline
        \footnotesize{constr.~(mult.)} & $0.90\%$ & $17.30\%$ & $1.01$ & $42.10\%$ \\ \hline
        \footnotesize{typical (sing.)} & $0.00\%$ & $0.10\%$ & $1.00$ & $0.00\%$ \\ \hline
        \footnotesize{typical (mult.)} & $0.00\%$ & $11.40\%$ & $1.00$ & $70.20\%$ \\ \hline
    \end{tabular}
    \caption{Estimator Performance ($0.125$ offset)}
    \label{tab:results2}
\end{table}

\begin{table}
    \centering
    \begin{tabular}{|l|r|r|r|r|}
    \hline
         & \footnotesize{arm error} &   \footnotesize{node error} & \footnotesize{error size} & \footnotesize{positivity} \\ \hline
        \footnotesize{unconstr.~(sing.)} & $0.30\%$ & $9.30\%$ & $1.00$ & $82.80\%$ \\ \hline
        \footnotesize{unconstr.~(mult.)} & $2.70\%$ & $42.90\%$ & $1.29$ & $4.60\%$ \\ \hline
        \footnotesize{constr.~(sing.)} & $0.00\%$ & $5.70\%$ & $1.00$ & $93.00\%$ \\ \hline
        \footnotesize{constr.~(mult.)} & $0.80\%$ & $18.30\%$ & $1.00$ & $31.50\%$ \\ \hline
        \footnotesize{typical (sing.)} & $0.00\%$ & $0.60\%$ & $1.00$ & $0.00\%$ \\ \hline
        \footnotesize{typical (mult.)} & $0.00\%$ & $12.50\%$ & $1.00$ & $55.20\%$ \\ \hline
    \end{tabular}
    \caption{Estimator Performance ($0.16$ offset)}
    \label{tab:results3}
\end{table}
\begin{table}
    \centering
    \begin{tabular}{|l|r|r|r|r|}
    \hline
         & \footnotesize{arm error} &   \footnotesize{node error} & \footnotesize{error size} & \footnotesize{positivity} \\ \hline
        \footnotesize{unconstr.~(sing.)} & $0.40\%$ & $5.20\%$ & $1.00$ & $59.60\%$ \\ \hline
        \footnotesize{unconstr.~(mult.)} & $4.00\%$ & $47.00\%$ & $1.28$ & $0.70\%$ \\ \hline
        \footnotesize{constr.~(sing.)} & $0.00\%$ & $2.60\%$ & $1.00$ & $69.20\%$ \\ \hline
        \footnotesize{constr.~(mult.)} & $1.10\%$ & $22.50\%$ & $1.02$ & $10.40\%$ \\ \hline
        \footnotesize{typical (sing.)} & $0.00\%$ & $0.40\%$ & $1.00$ & $0.00\%$ \\ \hline
        \footnotesize{typical (mult.)} & $0.00\%$ & $13.30\%$ & $1.00$ & $15.80\%$ \\ \hline
    \end{tabular}
    \caption{Estimator Performance ($0.25$ offset)}
    \label{tab:results4}
\end{table}

\begin{table}
    \centering
    \begin{tabular}{|l|r|r|r|r|}
    \hline
         & \footnotesize{arm error} &   \footnotesize{node error} & \footnotesize{error size} & \footnotesize{positivity} \\ \hline
        \footnotesize{unconstr.~(sing.)} & $0.50\%$ & $11.10\%$ & $1.00$ & $3.60\%$ \\ \hline
        \footnotesize{unconstr.~(mult.)} & $2.70\%$ & $67.50\%$ & $1.32$ & $0.20\%$ \\ \hline
        \footnotesize{constr.~(sing.)} & $0.30\%$ & $11.70\%$ & $1.00$ & $0.00\%$ \\ \hline
        \footnotesize{constr.~(mult.)} & $0.90\%$ & $44.60\%$ & $1.05$ & $0.20\%$ \\ \hline
        \footnotesize{typical (sing.)} & $0.00\%$ & $0.10\%$ & $1.00$ & $0.00\%$ \\ \hline
        \footnotesize{typical (mult.)} & $0.00\%$ & $38.20\%$ & $1.03$ & $0.00\%$ \\ \hline
    \end{tabular}
    \caption{Estimator Performance ($0.5$ offset)}
    \label{tab:results5}
\end{table}

In some cases, the listed ``errors" are not truly errors---merely ties broken differently---but we report them anyway. For the constrained and unconstrained single cases, all observed arm errors are either (a) when there are two arms of length $ n $ and $ n+1 $ respectively, so the source is equally likely be the central hub or one spot out on the longer arm, or (b) when two arms are tied for longest. In the typical single case, all observed node errors occur are when there are two arms, and the total number of overlaps is odd, so it is equally likely for the source to have been on either side of the center overlap. Similar non-errors may arise in other cases, but they are not the only (or even dominant) contributors to those error rates.

For multiple overlaps, arm errors tend to occur when the two largest  weighted lengths are very close in value, and due to granularity the observed maximum on the shorter arm is actually higher than on the longer arm.  The frequency of such errors is relatively independent of the offset value.  Further, arm length constraints reduce the incidence of such cases (as it becomes less likely the top two weighted arm lengths are close together), and they disappear entirely under typicality.

Node errors are more frequent than in the single case, as larger overlaps correspond to smaller increments in the weighed length.  Errors are most common when there is a wide discrepancy in the weighted lengths of the non-longest arms---a circumstance less likely under constraints, and even less so under typicality.  Unsurprisingly, lower error rates are observed when using offsets that yield positivity rates close to 50\%, i.e., biased outcomes tend to be less accurate.
Among the offsets examined, $0.25$ provides the best performance for single overlap settings, whereas $0.125$ seems to work best overall for multiple overlap settings.  This makes sense, since we are essentially trying to account for a difference of a single hop, and multiple overlaps correspond to shorter weighted lengths per hop on average than single hops.

\section{Conclusion}
\label{sec:conclusion}
By generalizing the extended star network to hypertrees, we leverage a known result for maximum likelihood source estimation to a much wider set of cases. In particular, we now  consider superspreading events, as well as nontrivial interactions between overlapping social bubbles.  Both of these social/epidemiological model enhancements are 
significant in understanding the propagation of contagion processes such as COVID-19. 
We have shown our approach is a solid method to address such challenges, preserves the optimality of the tree-based method, and has some desirable robustness properties as well.  As extensions, we aim to study settings with multiple sources \cite{JiTV2017} and with general hypergraphs \cite{JiTT2019}.

\bibliographystyle{IEEEbib}
\bibliography{refs}

\end{document}